\documentclass[sigconf,nonacm]{acmart}
\AtBeginDocument{%
  \providecommand\BibTeX{{%
    \normalfont B\kern-0.5em{\scshape i\kern-0.25em b}\kern-0.8em\TeX}}}

\usepackage{amsmath,amsthm,amsfonts}
\usepackage[utf8]{inputenc}
\usepackage[T1]{fontenc}

\usepackage{mathrsfs}  

\usepackage{multirow}

\usepackage{xcolor}
\usepackage[onelanguage,english,ruled,vlined,linesnumbered]{algorithm2e}

\theoremstyle{acmplain}
\newtheorem{theorem}{Theorem}[section]

\newtheorem{proposition}[theorem]{Proposition}

\theoremstyle{acmdefinition}

\newtheorem{assumption}[theorem]{Assumption}
\theoremstyle{remark}
\newtheorem{remark}[theorem]{Remark}

\newcommand{\Fq}{\mathbb{F}_q}

\newcommand{\N}{\mathbb{N}}

\renewcommand{\SS}{\mathcal{S}}

\newcommand{\MM}{\mathcal{M}}

\newcommand{\LL}{\mathcal{L}}

\renewcommand{\v}{\boldsymbol{v}}

\newcommand{\p}{\boldsymbol{p}}

\newcommand{\bvarphi}{\boldsymbol{\varphi}}

\newcommand{\eps}{\boldsymbol{\varepsilon}}
\newcommand{\e}{\boldsymbol{e}}
\newcommand{\y}{\boldsymbol{y}}
\renewcommand{\b}{\boldsymbol{b}}

\newcommand{\bz}{\boldsymbol{0}}

\newcommand{\rank}{\textnormal{rank}}

\newcommand{\ceil}[1]{\left\lceil #1 \right\rceil}
\newcommand{\floor}[1]{\left\lfloor #1 \right\rfloor}
\def\ie {\textit{i.e.\ }}

\def\iff{\textit{iff}}

\copyrightyear{2021}
\acmYear{2021}
\setcopyright{acmcopyright}\acmConference[ISSAC '21]{International Symposium on Symbolic and Algebraic Computation}{July 20--22, 2021}{Saint Petersburg, Russia}
\acmBooktitle{International Symposium on Symbolic and Algebraic Computation (ISSAC '21), July 20--22, 2021, Saint Petersburg, Russia}
\acmPrice{}
\acmDOI{}
\acmISBN{}

\begin{document}

\title{{Polynomial Linear System Solving with Random Errors: new bounds and early termination technique}}

\author{Eleonora Guerrini}
\email{guerrini@lirmm.fr}
\affiliation{%
%  \institution{LIRMM, U.niversit\'e de Montpellier, CNRS}
  \institution{LIRMM, U. Montpellier, CNRS}
  \city{Montpellier}
  \state{France}
  }

\author{Romain Lebreton}
\email{lebreton@lirmm.fr}
\affiliation{%
%  \institution{LIRMM, U.niversit\'e de Montpellier, CNRS}
  \institution{LIRMM, U. Montpellier, CNRS}
	\city{Montpellier}
	\state{France}
}

\author{Ilaria Zappatore}
\email{zappatore@inria.fr}
\affiliation{%
%	\institution{Inria,\\ Laboratoire d'informatique de l'\'Ecole polytechnique\\ (LIX, UMR 7161)}
	\institution{Inria, LIX}
	\city{Palaiseau}
	\state{France}
}

\begin{abstract}
  This paper deals with the polynomial linear system solving with errors (PLSwE)
  problem. Specifically, we focus on the evaluation-interpolation technique for
  solving polynomial linear systems and we assume that errors can occur in the
  evaluation step. In this framework, the number of evaluations needed to
  recover the solution of the linear system is crucial since it affects the
  number of computations. It depends on the parameters of the linear system
  (degrees, size) and on a bound on the number of errors.

  Our work is part of a series of papers about PLSwE aiming to reduce this
  number of evaluations. We proved in [Guerrini et al., Proc. ISIT'19] that if
  errors are randomly distributed, the bound of the number of evaluations can be
  lowered for large error rate.

  In this paper, following the approach of [Kaltofen et al., Proc. ISSAC'17], we
  improve the results of [Guerrini et al., Proc. ISIT'19] in two
  directions. First, we propose a new bound of the number of evaluations,
  lowering the dependency on the parameters of the linear system, based on work
  of [Cabay, Proc. SYMSAC'71].  Second, we introduce an early termination
  strategy in order to handle the unnecessary increase of the number of
  evaluations due to overestimation of the parameters of the system and on the
  bound on the number of errors.
\end{abstract}

\maketitle

\section{Introduction}
 
Solving polynomial linear systems (PLS) of the form $A(x)\y(x)=\b(x)$ where A is a nonsingular square matrix and $\b$ is a vector of polynomials over a finite field $\Fq$ is a classical computer algebra problem. The solution $\y(x)$ is a vector of rational functions.
This problem can be efficiently solved by parallelizing the classical evaluation-interpolation technique considering a network of $L$ nodes that independently compute the evaluations $A(\alpha_j)$ and $\boldsymbol{b}(\alpha_j)$ at a given evaluation point $\alpha_j\in \Fq$ and the solution of the evaluated system $\boldsymbol{y}_j=A(\alpha_j)^{-1}\b(\alpha_j)$. 
The nodes then send the so-obtained $\y_1, \ldots, \y_L$ to the master node which finally performs a Cauchy interpolation to recover the solution $\y(x)$. 
As in \cite{boyer_numerical_2014, kaltofen_early_2017}, this paper focuses on a scenario in which the nodes could make errors, possibly computing  $\boldsymbol{y}_j\neq A(\alpha_j)^{-1}\b(\alpha_j)$. After receiving all these evaluations, the master node performs a \textit{Cauchy interpolation with errors} in order to recover the solution $\y(x)$. The problem that the master node has to face, \ie recovering the solution $\y(x)$ of the PLS given its evaluations, some of which erroneous, is what we call Polynomial Linear System Solving with errors (PLSwE).

In order to solve PLSwE, we can exploit decoding techniques of Reed-Solomon (RS)
codes, as shown in \cite{boyer_numerical_2014, kaltofen_early_2017}. Basically
they set out a system of linear equations ({\it key equation}) (as the
Welch-Berlekamp decoding method, \cite{berlekamp_error_1986}) and bound $L$ (\ie
the number of nodes, which coincides with the number of evaluation points
$\alpha_j$ in some way to guarantee the uniqueness of the solution. The goal is
to minimize the number $L$ of evaluation points needed to recover the solution
or equivalently to maximize the bound on the number of errors (\textit{decoding
  radius}) that we could correct. In \cite{boyer_numerical_2014,
  kaltofen_early_2017}, as for classical RS codes they can correct up to the
{\it unique decoding radius}.

In \cite{kaltofen_early_2017} it is shown that there are mainly two ways to
bound $L$: first, considering the problem as a generalized rational function
reconstruction (RFR) and analyze it in terms of some estimations of the solution
degree. Second, exploit the linear algebra structure of the problem taking also
into account the degrees of the input matrix $A$ and of the vector $b$.  More
recently,  \cite{guerrini_polynomial_2019} presented an algorithm that
corrects errors beyond the unique decoding radius (equivalently with less
evaluation points than \cite{boyer_numerical_2014, kaltofen_early_2017})
recovering the solution for almost all errors. The idea is to remark that the
PLSwE problem can be viewed as a generalization of the decoding of the
Interleaved Reed Solomon codes (IRS).

IRS can be seen as the simultaneous evaluation of a vector of polynomials.

Results from decoding IRS codes show that, if errors are uniformly distributed,
the larger the dimension of the vector is, the more errors we can correct,
exceeding the standard unique decoding radius (see
\cite{bleichenbacher_decoding_2003,brown_probabilistic_2004,schmidt_enhancing_2007,
  schmidt_collaborative_2009,schmidt_syndrome_2010}), asymptotically reaching
the optimal error capability of the Shannon bound
(\cite{shannon_communication_1948}).

A first contribution of this work consists in the combination of the
advantages of IRS decoding techniques from \cite{guerrini_polynomial_2019} with
the counting of \cite{kaltofen_early_2017} which exploits the linear algebra
setting as in \cite{cabay_exact_1971} (see Section~\ref{sec:newbound}).

\smallskip Recall that our goal is to lower the number of evaluations in order
to reduce the nodes computations, at the expense of potentially increasing the
complexity of interpolation by the master node.

All the bounds of the number of evaluations introduced for PLSwE solving depend
on some upper bounds on the degree of the solution $\y(x)$ that we want to
recover and on the number of errors.  These upper bounds could overestimate the
actual degrees of $\y(x)$ and number of errors. The discrepancy between these
quantities may significantly overestimate the number of evaluations needed for
the computations compared to the actual number needed to recover the solution.
We propose an \textit{early termination technique} (as in
\cite{kaltofen_early_2017}), an adaptive strategy which, starting from a
\textit{minimal value} of evaluation points, iteratively increments this number
until a nontrivial result is found.

In Section~\ref{sec:fixed-error-bound} we present an early termination technique
for a fixed bound $\tau$ on the number of errors, which is the classical error
correcting codes framework. However, the number of errors could grow with the
number of evaluations $L$, which is gradually incremented in the early
termination technique. For this purpose, we present in
Section~\ref{sec:linearBound} a scenario in which the error bound linearly
depends on $L$.  Compared to the early termination techniques of
\cite{kaltofen_early_2017}, we decrease the number of evaluation points. In
return, our algorithm may fail for a small fraction of errors; we give an
estimation of the success probability of our algorithm in presence of random
errors. To sum up, our second contribution is to propose an early termination
strategy which benefits from the IRS decoding approach, is sensitive to the real
number of errors, and adapts to linear error bound.  To the best of our
knowledge, the dependency on the real number of errors is original in the
literature.

\smallskip
The paper is organized as follows: in Section~\ref{sec:PLSwE} we recall the scenario of PLSwE with results revisited from literature, in Section~\ref{sec:newbound} we present a new bound of the number of evaluation points needed for PLSwE solving in presence of random errors  and finally in Section~\ref{sec:ET} we introduce an early termination algorithm that succeeds for almost all errors.
\section{Polynomial Linear System Solving with Errors}{\label{sec:PLSwE}}

Let $\Fq$ a finite field of order $q$. 
Consider a polynomial linear system (PLS),
\begin{equation}\label{eq:isitPLS}
A(x)\y(x)=\b(x)
\end{equation}
where $A\in \Fq[x]^{n\times n}$ is nonsingular and $\b\in \Fq[x]^{n \times 1}$. This system admits only one solution $\y=\frac{\v}{d}\in \Fq(x)^{n \times 1}$, \ie a vector of rational functions with the same denominator. We assume that $\gcd(\gcd_i(v_i),d)=1$ and that $d$ is monic.

The evaluation-interpolation \cite{mcclellan_exact_1977} is a classic technique for solving PLS. It consists in 
\begin{itemize}
	\item \textit{(evaluation)} the evaluation of $A$ and $\b$ at $L\leq q$ distinct evaluation points $\{\alpha_1, \ldots, \alpha_L\}$.
	In this work, for simplicity we omit the \textit{rank drop case} study, \ie we suppose that for any $\alpha_j$ the corresponding evaluated matrix $A(\alpha_j)$ is still full rank. All the results of this work can be extended to the general case of rank drops (as in \cite{kaltofen_early_2017} for more details).  
	\item \textit{(Pointwise interpolation of the evaluated systems)} Compute $\y_j=A(\alpha_j)^{-1}\b(\alpha_j)=\frac{\v(\alpha_j)}{d(\alpha_j)}$, for any $j$.
	\item {\textit{(Interpolation)}} Reconstruct $(\v,d)\in \Fq[x]^{(n+1) \times 1}$, given the evaluated solutions $\y_j$ for any $j$ and the degree bounds $N>\deg(\v):=\max_{1\leq i \leq n}\{\deg(v_i)\}$ and $D>\deg(d)$.
\end{itemize}
 In order to minimize the  number of evaluation points needed to \textit{uniquely recover} the solution one can consider  
\begin{equation}{\label{eq:numbEvPointsPLS}}
\LL(N,D):=\min\{\underbrace{N+D-1}_{\LL_{RFR}}, \underbrace{\max\{\deg(A)+N, \deg(\b)+D\}}_{\LL_{PLS}}\},
\end{equation}
where $\deg(A):=\max_{1\leq i,j\leq n}\deg(a_{i,j}(x))$.
We use this notation to stress out the dependency of the degree bounds $N$ and $D$.
Recall that $\LL_{RFR}$ is the minimum number of evaluation points needed to uniquely interpolate a \textit{rational function}  (\ie the \textit{Cauchy interpolation} problem) \cite[Section 5.7]{vzgathen_gerhard_2013}.
On the other hand, $\LL_{PLS}$ is the minimum number of evaluation points needed to uniquely recover a \textit{rational function which is a solution of a PLS} (\cite{cabay_exact_1971}). Indeed, in this case we also assume to know the degrees of $A,\b$ or  their upper bound.

We remark that if the bounds are tight, \ie $\deg(\v)+1=N$ and $\deg(d)+1=D$, then $L_{PLS}<L_{RFR}$ if and only if $\deg(d)>\deg(A)$ (see \cite[Theorem~3.1]{kaltofen_early_2017}).
In the following section we formalize and describe our error scenario.

\subsection{Parallelization and error model}\label{subsection:ErrorModel}
Fix $L$ pairwise distinct evaluation points $\{\alpha_1,\ldots, \alpha_{L}\}$.
Assume that any node, computes $\y_j=A(\alpha_j)^{-1}\b(\alpha_j)\in \Fq^{n \times 1}$ for any $1\leq j \leq L$. 

These nodes could make some errors and compute $\y_j\neq A(\alpha_j)^{-1}\b(\alpha_j)=\frac{\v(\alpha_j)}{d(\alpha_j)}$. Notice that, in this error model, the number of errors coincides with the number of nodes which compute an incorrect result.
We assume that after the nodes computations the master node receives the following matrix
\begin{equation}{\label{eq:nodesErrResult}}
Y=\left(\frac{\v(\alpha_1)}{d(\alpha_1)}, \ldots, \frac{\v(\alpha_{L})}{d(\alpha_{L})}\right)+\Xi
\end{equation}
where $\Xi\in \Fq^{n \times L}$ is the error matrix. We denote by $\Xi_{*,j}$ the $j$th column of $\Xi$. The \textit{error support} of the error matrix is $E:=\{j \mid \Xi_{*,j}\neq \boldsymbol{0}\}$.
After receiving $Y$, the master node has to recover the solution $(\v,d)$ of the PLS \eqref{eq:isitPLS}. In this work, we focus on this step, which we call \textit{polynomial linear system solving with errors} (PLSwE). 

Formally PLSwE refers to the problem of recovering $(\v,d)$ the solution of \eqref{eq:isitPLS} given,
\begin{itemize}

	\item $1\leq L \leq q$ distinct evaluation points in $\Fq$, $\{\alpha_1, \ldots, \alpha_L\}$,
	\item the degree bounds $1\leq N, D\leq L$, such that $N>\deg(\v)$, $D>\deg(d)$ and $\deg(A)$, $\deg(\b)$,
	\item an upper bound on the number of errors occurred at the parallelization step, \ie $\tau\geq |E|$, where $E=\{j \mid \Xi_{*, j}\neq \boldsymbol{0}\}$,
	\item the matrix $Y$ as in \eqref{eq:nodesErrResult}.

\end{itemize}

\subsection{Resolution method for PLSwE and previous results}
As in \cite{boyer_numerical_2014, kaltofen_early_2017, guerrini_polynomial_2019}, in order to solve PLSwE,  we search for solutions $(\bvarphi, \psi)=(\varphi_1, \ldots, \varphi_n, \psi)\in \Fq[x]^{(n+1)\times 1}$ of the following \textit{key equations}
\begin{equation}{\label{eq:keyEqPLSWE}}
\varphi_i(\alpha_j)=y_{i,j}\psi(\alpha_j), \hspace{0.1cm} \deg(\varphi_i)<N+\tau,\hspace{0.1cm} \deg(\psi)<D+\tau.
\end{equation}
for any $1\leq i \leq n$ and $1\leq j \leq L$.

Note that the key equations \eqref{eq:keyEqPLSWE} are the vector generalization of the classic computer algebra problem of the \textit{Cauchy interpolation} (\cite[Section 5.7]{vzgathen_gerhard_2013}).

This approach is the generalization of the  Welch-Berlekamp decoding method \cite{berlekamp_error_1986} for Reed-Solomon codes.
\smallskip
In this framework, it is crucial to determine the smallest number of evaluation points $L$ needed to guarantee the \textit{uniqueness} of a solution of these key equations, in the sense that we explain in what follows.

We denote
$$\SS_{Y,N+\tau,D+\tau}:=\{(\bvarphi, \psi)\in \Fq[x]^{(n+1)\times 1}  \textit{ satisfying \eqref{eq:keyEqPLSWE}}\}.$$

We now consider the $\Fq[x]$-module $\MM$ spanned by the solutions of the key
equations \eqref{eq:keyEqPLSWE}.
More specifically, any element of $\MM$ is a linear combination with polynomial coefficients of $\rank(\MM)$ solutions $(\bvarphi, \psi_i)\in \SS_{Y,N+\tau,D+\tau}$ (\ie a basis of $\MM$). 
The case in which $\MM$ is uniquely generated, \ie $\rank(\MM)=1$, corresponds to what we refer to \textit{uniqueness} of the solution of the key equations.
Indeed, in this case $\MM$ is generated by only one element $(\bvarphi, \psi)\in \SS_{Y,N+\tau,D+\tau}$, \ie for any $(\bvarphi', \psi')\in \MM$, there exists a polynomial $R\in \Fq[x]$ such that $(\bvarphi', \psi')=(R\bvarphi, R\psi)$.
In other terms, if the polynomials $\psi, \psi'$ are nonzero, the two vectors of rational functions $\bvarphi/\psi$ and $\bvarphi'/\psi'$ are equal. 

\begin{remark}\label{remark:loc}
	Let $\Lambda:=\prod_{j \in E}(x-\alpha_j)$ be the \textit{error locator polynomial}, \ie the monic polynomial of degree $\deg(\Lambda)=|E|$ whose roots are the \textit{erroneous evaluations}.
	
	We have that $(\Lambda \v, \Lambda d)\in \SS_{Y,N+\tau,D+\tau}$. Indeed, for any $1\leq i \leq n$ and $1\leq j\leq L$,
	$
	\Lambda(\alpha_j) v_i(\alpha_j)=y_{i,j}\Lambda(\alpha_j)d(\alpha_j)
	$
	and we also have that $\deg(\Lambda\v)<N+\tau$ and $\deg(\Lambda d)<D+\tau$.
\end{remark}

In \cite{boyer_numerical_2014, kaltofen_early_2017} is provided the minimum number of points which guarantees the uniqueness of the solutions of the key equations \eqref{eq:keyEqPLSWE}. 
The following proposition is a restatement of this result using  definitions and notations of this paper.
  We will later prove this result in a more \textit{general} context (see Proposition~\ref{prop:stopcritkalto}.
\begin{proposition}\label{prop:kaltoNbpointsUniqueness}
	If $
	L\geq L_{KPSW}:= \LL(N+\tau,D+\tau)+\tau$ (see \eqref{eq:numbEvPointsPLS}),
	then the rank of $\MM$, the $\Fq[x]$-module generated by solutions, is $1$. By Remark \ref{remark:loc}, if $(\bvarphi, \psi)$ is a generator of $\MM$ with $\psi$ monic, then $(\bvarphi, \psi)=(\Lambda \v, \Lambda d)$.
\end{proposition}

This proposition tells us that if $L\geq \LL(N+\tau,D+\tau)+\tau=\min\{N+D-1, \max\{\deg(A)+N, \deg(\b)+D\}\}+2\tau$, then $\rank(\MM)=1$ and the solution space is spanned by vectors of the form $(x^i\Lambda \v,x^i\Lambda d)$, \ie
$
\SS_{Y,N+\tau,D+\tau}=\langle x^i\Lambda \v,x^i\Lambda d\rangle_{0\leq i <\delta_{N+\tau,D+\tau}}
$
where
\begin{equation}{\label{eq:deltaNDtau}}
  \delta_{N+\tau,D+\tau}:=\min\{N+\tau-(\deg(\v)+|E|), D+\tau-(\deg(d)+|E|)\}
\end{equation}
By convention, if $\delta_{N+\tau,D+\tau} \leq 0$ we set
$\langle x^i\Lambda \v, x^i\Lambda d\rangle_{0\leq i < \delta}
=\{(\boldsymbol{0},0)\}$.

Indeed, notice that
\begin{align*}
  \deg(x^i\Lambda \v)=i+|E|+\deg(\v)\leq N-1+\tau\\
  \deg(x^i\Lambda d)=i+|E|+\deg(d)\leq D-1+\tau
\end{align*}
and so
$i\leq \min\{N-1+\tau-(\deg(\v)+|E|), D-1+\tau-(\deg(d)+|E|)\}=\delta_{N+\tau,D+\tau}-1.$

\begin{remark}{\label{remark:coeffMatrIRS}}

  Let $\mathcal{V}_{m,p}=(\alpha_i^{j-1})_{\substack{1\leq i \leq m\\1\leq j \leq p}}$.
  Consider the homogeneous linear system related to \eqref{eq:keyEqPLSWE}. We observe that the set of solutions $\SS_{Y,N+\tau,D+\tau}$ is the kernel of the matrix
  \begin{equation}{\label{eq:coeffMatrixPLSwE}}
    M_{Y,N+\tau,D+\tau}=\left(
      \begin{array}{l l l | c}
        \mathcal{V}_{L,N+\tau} &        &	               & -D_1\mathcal{V}_{L, D+\tau}\\
                               & \ddots	&                      & \vdots\\
                               &	& \mathcal{V}_{L,N+\tau} & -D_n\mathcal{V}_{L, D+\tau}\\
      \end{array}\right)
  \end{equation}
  where for any $1\leq i \leq n$, $D_i$ is the diagonal matrix whose elements
  on the diagonal are $y_{i,1},\ldots,y_{i,L}$.
  
  In \cite{boyer_numerical_2014, kaltofen_early_2017} was proposed an
  algorithm which computes a column echelon form of $M_{Y,N+\tau,D+\tau}$ in
  order to find the minimal degree solution of $\SS_{Y,N+\tau,D+\tau}$, \ie
  $(\Lambda \v, \Lambda d)$.  Another approach to find this solution could be to
  compute a basis (as for instance in \cite{olesh_vector_2007,
    rosenkilde-algorithms-sim_2016}) of the $\Fq[x]$-module $\MM$. Indeed, we
  have seen in Proposition~\ref{prop:kaltoNbpointsUniqueness} that if
  $L\geq L_{KPSW}$, the module is generated by $(\Lambda \v, \Lambda d)$.

  Note that we can recover $(\v,d)$ from $(\Lambda \v, \Lambda d)$ by dividing
  by $\Lambda=\gcd(\Lambda \v, \Lambda d)$. We denote by
  \texttt{FindSolution}$(Y,N+\tau,D+\tau)$ the algorithm that computes $(\v, d)$
  from $\SS_{Y,N+\tau,D+\tau}$ using one of the above methods, followed by the
  division by the $gcd$.

\end{remark} 

	We now observe that if $d\in \Fq$ (\ie$D=1$), PLSwE is the problem of recovering a \textit{vector} of $n$ \textit{polynomials}, given its evaluations, some of which erroneous. This problem can be viewed as the problem of decoding $n$ codewords of a Reed-Solomon code of length $L$ and dimension $N$. 
	In this specific case, Proposition~\ref{prop:kaltoNbpointsUniqueness} tells us that with
	$
	L\geq N+2\tau-1 
	$
	or equivalently if $|E|\leq \tau\leq \frac{L-N}{2}=:\tau_0$ we can uniquely decode the $n-$ vector of RS codewords. From a coding theory point of view, $\tau_0$ the \textit{unique decoding radius} of an RS code of length $L$ and dimension $N$.
	We now recall that an $n$-Interleaved RS code of length $L$ and dimension $N$ is the direct sum of $n$-RS codes with the same length and dimension. An $n$-IRS codeword is the evaluation of a vector of polynomials at $L$ distinct evaluation points. Therefore, PLSwE with constant $d$ can also be seen as the decoding of an $n$-IRS codeword. The advantage of considering the problem under this interleaving point of view consists in the fact that we can extend the results of the decoding of IRS (\cite{bleichenbacher_decoding_2003, brown_probabilistic_2004, schmidt_enhancing_2007, schmidt_collaborative_2009, schmidt_syndrome_2010} and correct beyond the unique decoding radius (or equivalently reduce the number of evaluations needed to recover the solution of our PLS).

Indeed, in \cite{guerrini_polynomial_2019} we proved that in the general case of PLSwE, we can reconstruct the solution $(\v,d)$ of the PLS with
$
L_{GLZ19}:=N+\deg(d)-1+|E|+\left \lceil\frac{|E|}{n}\right\rceil
$
evaluation points for almost all errors.
In this result we assume to know exactly the actual degree of the denominator $d$ and the actual number of errors $|E|=|\{j \mid \Xi_{j,*}\neq \boldsymbol{0}\}|$. Notice that this assumption is quite strong. For this reason in this work we introduce a new bound on $L$ which generalizes $L_{GLZ19}$  by assuming to know some upper bounds on the degree of $d$ and of $|E|$. Our new bound also takes into account the linear algebra setting of the problem (see Equation~(\ref{eq:numbEvPointsPLS})).

\section{New bound for PLSwE}\label{sec:newbound}
There are two new contributions in this section. First, we relax the constraint of  \cite{guerrini_polynomial_2019} introducing a number of evaluations which only depends on some upper bounds on the degree of the denominator and on the number of errors. We also introduce another independent counting on the number of evaluations that takes into account $\deg(A), \deg(\b)$ of the PLS \eqref{eq:isitPLS} as in \cite{cabay_exact_1971,kaltofen_early_2017}. 
We prove that with
\begin{equation}{\label{eq:LGLZ}}
L\geq L_{GLZ}:=\LL(N+\tau,D+\tau)+\left \lceil\frac{\tau}{n}\right \rceil  
\end{equation}
evaluation points, where $\LL(N+\tau,D+\tau)$ is defined as in \eqref{eq:numbEvPointsPLS},
we can uniquely reconstruct the solution for almost all errors (Theorem~\ref{thm:ISIT2}).

\begin{theorem}{\label{thm:ISIT2}}
  Let $\tau\geq 0$ and $n, N, D\geq 1$. Let $ L\geq L_{GLZ} $, consider the set
  of evaluation points $\{\alpha_1, \ldots, \alpha_{L}\}$ and
  $E\subseteq\{1, \ldots, L\}$, with $|E|\leq \tau$.  Moreover, fix
  $A(x)\y(x)=\b(x)$ and denote $\y(x) = \frac{\v(x)}{d(x)}$ with
  $\gcd(\gcd_i(v_i),d)=1$ and $d$ monic. Let $\deg(\v)<N$ and
  $\deg(d)<D$.
	
	Consider the random matrix $Y$, where we denote by $\y_j:=Y_{*,j}$ for any $1\leq j \leq L$, constructed as follows:
	\begin{itemize}
		\item if $j\in E$, $\y_j$ is a uniformly distributed element of $\Fq^{n\times 1}$,
		\item if $j \notin E$, $\y_j=\frac{\v(\alpha_j)}{d(\alpha_j)}$,
	\end{itemize}
	then
  
	$\SS_{Y,N+\tau,D+\tau}=\langle x^i\Lambda \v, x^i\Lambda d\rangle_{0\leq i < \delta_{N+\tau,D+\tau}},$
	where $\delta_{N
	+\tau,D+\tau}$ defined as in \eqref{eq:deltaNDtau},
	with probability at least $1-\frac{D+\tau}{q}$.
\end{theorem}

\begin{proof}
  
  First notice that $(\Lambda \v, \Lambda d)\in \SS_{Y,N+\tau,D+\tau}$ and so \linebreak
  %\begin{equation*}{\label{eq:inclusionProofMainTheorem}}
    $\langle x^i\Lambda \v, x^i\Lambda d\rangle_{0\leq i < \delta_{N+\tau,D+\tau}} \subseteq \ker(M_{Y,N+\tau,D+\tau})
    =\SS_{Y,N+\tau,D+\tau}$
  %\end{equation*}
  where the last equality was already observed in the Remark~\ref{remark:coeffMatrIRS}.
  The proof is based on the following two steps:
  
  \begin{enumerate}
  \item\label{item1proof} we prove that there exists a draw of columns of $Y$,  $\y_j$ for $j \in E$, for which the corresponding solution space $\SS_{Y,N+\tau,D+\tau}=\langle x^i\Lambda \v, x^i \Lambda d\rangle_{0\leq i < \delta_{N+\tau,D+\tau}}$. Notice that we only need to prove this inclusion $\subset$ since the other one is always true.
  \item\label{item2proof} In the second part, we derive the bound on fraction of errors for which the solution space is not of the form $\langle x^i\Lambda \v, x^i \Lambda d\rangle$.
  \end{enumerate}
  
  \ref{item1proof}) First assume that $L_{GLZ}=N+D-1+\tau+\left\lceil \frac{\tau}{n}\right \rceil$. Consider a partition of the set of error positions $E$, \ie $E=\cup_{i=1}^nI_i$, such that for any $1\leq i \leq n$, $|I_i|\leq \lceil |E|/n\rceil$.
  Note that such a partition exists since $n\lceil |E|/n\rceil \geq |E|$. For any $j \in E$, denote by $i_j$ the unique index such that $j\in I_{i_j}$.
  
  Construct a matrix $V$, such that $V_{*,j}=\frac{\v(\alpha_j)}{d(\alpha_j)}$
  if $j \not\in E$, and for $j \in E$ consider $V_{*,j} \in \Fq^{n \times 1}$ chosen so that
  $
  \v(\alpha_j)-d(\alpha_j)V_{*,j}=\eps_{i_j}
  $
 (where $\eps_{i}$ is the $i$th element of the canonical basis of $\Fq^{n \times 1}$).  

  Consider $(\bvarphi, \psi)\in \SS_{V,N+\tau,D+\tau}$ and multiply $
  \v(\alpha_j)-d(\alpha_j)V_{*,j}=\eps_{i_j}
  $by  $\psi(\alpha_j)$, for $j \in E$. Since $(\bvarphi, \psi)\in \SS_{V,N+\tau,D+\tau}$, by the key equations
  \eqref{eq:keyEqPLSWE}, we get
  $\psi(\alpha_j)\eps_{i_j}=\psi(\alpha_j)\v(\alpha_j)-d(\alpha_j)\psi(\alpha_j)V_{*,j}=(\psi\v-\bvarphi)(\alpha_j).$

  Fix $1\leq i \leq n$, we claim that for any $j \notin I_i$ then
  $\psi(\alpha_j)v_i(\alpha_j)-d(\alpha_j)\varphi_i(\alpha_j)=0$.  Indeed, if
  $j \notin E$, then $V_{*,j}=\frac{\v(\alpha_j)}{d(\alpha_j)}$ and so by
  replacing $V_{*,j}$ in the key equations \eqref{eq:keyEqPLSWE}, we have $\psi(\alpha_j)v_i(\alpha_j)-d(\alpha_j)\varphi_i(\alpha_j)=0$. Now
  if $j \in E\setminus I_i$, by the choice of $V_{*,j}$, then
  $\psi(\alpha_j)v_i(\alpha_j)-d(\alpha_j)\varphi_i(\alpha_j)=0$.

 Note that for any $1\leq i \leq n$, $\deg(\psi v_i-d\varphi_i)<N+D+\tau-1$.
  On the other hand the number of roots of these polynomials are $L-|I_i|\geq L-\lceil |E|/n\rceil \geq L_{GLZ} -\tau/n$ and since $L_{GLZ}\geq  N+D-1+\tau+\tau/n$ it is then $L-|I_i|\geq N+D+\tau-1$. Therefore, since all these polynomials have more roots than their degree they are the zero polynomials.
  Hence, $\psi(x)\v(x)-d(x)\bvarphi(x)=\boldsymbol{0}$.
  
  \smallskip
  On the other hand, assume that $L_{GLZ}=\max\{\deg(A)+N,\deg(\b)+D\}+\left\lceil\frac{\tau}{n}\right\rceil+\tau$.
  As before, we can consider a partition of $E$, $E=\cup_{i=1}^n I_i$, such that for any $1\leq i \leq n$, $|I_i|\leq \lceil |E|/n\rceil$.
  
  Construct a matrix $V$, such that $V_{*,j}=\frac{\v(\alpha_j)}{d(\alpha_j)}$
  if $j \notin E$, and so that $V_{*,j}$ satisfies
  $\v(\alpha_j)-d(\alpha_j)V_{*,j}=-A(\alpha_j)^{-1}d(\alpha_j)\eps_{i_j}$ when
  $j \in E$.
 
 For $j \in E$,
  $\v(\alpha_j)-d(\alpha_j)V_{*,j}=-A(\alpha_j)^{-1}d(\alpha_j)\eps_{i_j}$ or equivalently 
  $A(\alpha_j)d(\alpha_j)V_{*,j}-A(\alpha_j)\v(\alpha_j)=d(\alpha_j)\eps_{i_j}$. Hence,
  $\eps_{i_j}= A(\alpha_j)V_{*,j}-A(\alpha_j)\frac{\v(\alpha_j)}{d(\alpha_j)}=A(\alpha_j)V_{*,j}- \b(\alpha_j)$. Notice that by assumption $d(\alpha_j)\neq 0$.

  By multiplying by $\psi(\alpha_j)$ we get
  $
  A(\alpha_j)V_{*,j}\psi(\alpha_j)-\b(\alpha_j)\psi(\alpha_j)=\psi(\alpha_j)\eps_{i_j}
  $ and since $(\bvarphi, \psi)\in \SS_{V,N,D,\tau}$, it satisfies
  $\bvarphi(\alpha_j)=V_{*,j}\psi(\alpha_j)$ and so we have
  $ (A\bvarphi-\b\psi)(\alpha_j)=\psi(\alpha_j)\eps_{i_j}.  $
  
  We now denote $ \p:=A(x)\bvarphi(x)-\psi(x)\b(x)\in \Fq[x]^{n\times 1}$. Fix
  $1\leq i \leq n$, we claim that for any $j \notin I_i$ then $p_i(\alpha_j)=0$,
  where $p_i$ is the $i$th component of $\p$.  Indeed, if $j \notin E$, then
  $V_{*,j}=\frac{\v(\alpha_j)}{d(\alpha_j)}=A(\alpha_j)^{-1}\b(\alpha_j)$ and so
  since $\bvarphi(\alpha_j)=V_{*,j}\psi(\alpha_j)$, we get
  $\p(\alpha_j)=A(\alpha_j)\bvarphi(\alpha_j)-\psi(\alpha_j)\b(\alpha_j)=0$. On the other hand, if $j \in E\setminus I_i$ then by the choice of
  $V_{*,j}$, then $p_i(\alpha_j)=0$.

  Therefore, for $1\leq i \leq n$ and $j \notin I_i$, then $p_i(\alpha_j)=0$. Note that
  $\deg(p_i(x))<\max\{\deg(A)+N, \deg(\b)+D\}$.  On the other hand the roots of
  this polynomial are
  $L-|I_i| \geq L_{GLZ} -\tau/n = \max\{\deg(A)+N, \deg(\b)+D\}+\tau$. So we can
  conclude that $\p(x)=A(x)\bvarphi(x)-\psi(x)\b(x)=\boldsymbol{0}$.
  
  Now, since $A(x)\v(x)=d(x)\b(x)$ if we multiply this equation by $\psi(x)$ and also $\p(x)$ by $d(x)$ and we subtract both the equations we finally get 
  $\bvarphi(x)d(x)-\psi(x)\v(x)=\boldsymbol{0}$.
  
  \smallskip
  Therefore, in both cases we have
  $\bvarphi(x)d(x)-\psi(x)\v(x)=\boldsymbol{0}$.  Now, since $\y(x)=\frac{\v(x)}{d(x)}$ and $\gcd(\gcd_i(v_i),d)=1$ and $d$ is monic, there exists $R \in \Fq[x]$ such that $\bvarphi=R\v$ and $\psi=Rd$.
  Notice that for any $1\leq j \leq L$ by the key equations
  \eqref{eq:keyEqPLSWE} we get,
  $
  0=\bvarphi(\alpha_j)-\psi(\alpha_j)V_{*,j}=R(\alpha_j)[\v(\alpha_j)-V_{*,j}d(\alpha_j)].
  $ By construction, if $j\in E$, then
  $\v(\alpha_j)-V_{*,j}d(\alpha_j)\neq \boldsymbol{0}$ and so
  $R(\alpha_j)=0$. Therefore, the error locator polynomial
  $\Lambda =\prod_{j\in E}(x-\alpha_j)$ divides $R$ and so
  $(\bvarphi,\psi)\in \langle x^i\Lambda \v, x^i\Lambda d\rangle_{0\leq i <
    \delta_{N+\tau,D+\tau}}$.  Hence,
  $\SS_{V,N+\tau,D+\tau}\subseteq \langle x^i\Lambda \v, x^i\Lambda
  d\rangle_{0\leq i < \delta_{N+\tau,D+\tau}}$ and so the equality holds.
  
  Hence, we finally get
  $\SS_{V,N+\tau,D+\tau}=\langle x^i\Lambda \v, x^i\Lambda d\rangle_{0\leq i < \delta_{N+\tau,D+\tau}}$ for a draw $V$ of $Y$.
  
  \bigskip
  \ref{item2proof}) We now conclude the proof by estimating the fraction of errors for which the solution space is exactly of the form $\SS_{Y,N+\tau,D+\tau}=\langle x^i\Lambda \v, x^i\Lambda d\rangle_{0\leq i < \delta_{N+\tau,D+\tau}}$.
  Now for a generic instance of $Y$ recall that $\langle x^i\Lambda \v, x^i\Lambda d\rangle_{0\leq i < \delta_{N+\tau,D+\tau}}\subseteq \SS_{Y,N+\tau,D+\tau}=\ker(M_{Y,N+\tau,D+\tau}),$ then $\dim(\ker(M_{Y,N+\tau,D+\tau}))\geq \delta_{N+\tau,D+\tau}$.
  By the Rank-Nullity Theorem we have that 
  $
  \rank(M_{Y,N+\tau,D+\tau})\leq n(N+\tau)+D+\tau-\delta_{N+\tau,D+\tau}=:\rho.
  $
  On the other hand, as proved above, there exists a draw $V_{*,j}$ of $Y_{*,j}$, for $j\in E$, such that $\rank(M_{V,N+\tau,D+\tau})=\rho$.
  This means that there exists a nonzero $\rho$-minor in $M_{V,N+\tau,D+\tau}$. We consider the same nonzero $\rho$-minor in $M_{Y,N+\tau,D+\tau}$ as a multivariate polynomial $C$ whose indeterminates are $(y_{i,j})_{\substack{1\leq i \leq n\\j\in E}}$. We remark that we showed the existence of a draw $V_{*,j}$ of $Y_{*,j}$, for $j\in E$, such that $C(V_{*,j})$ is non zero. Hence, the polynomial $C$ is nonzero.
  For any matrix $Y$ such that $(Y_{*,j})_{j\in E}$ is not a root of $C$, then $\SS_{Y,N+\tau,D+\tau}= \langle x^i\Lambda \v, x^i\Lambda d\rangle_{0\leq i < \delta_{N+\tau,D+\tau}}$. Note that the total degree of the polynomial $C$ is at most $D+\tau$, since only the last $D+\tau$ columns of the matrix $M_{Y,N+\tau,D+\tau}$ contains the variables $(y_{i,j})_{\substack{1\leq i \leq n\\j\in E}}$ (see Remark~\ref{remark:coeffMatrIRS}).
  
  Finally, by the Schwartz-Zippel Lemma, the polynomial $C$ cannot be zero in more than $(D+\tau)/q$ fractions of its domain. Therefore, we can conclude that the probability that \\ $\SS_{Y,N+\tau,D+\tau}\neq \langle x^i\Lambda \v, x^i\Lambda d\rangle_{0\leq i < \delta_{N+\tau,D+\tau}}$ is at most $(D+\tau)/q$.
\end{proof}

\section{Early Termination Strategy}\label{sec:ET}

All the bounds on the number of evaluations $L$ introduced so far, \ie
$L_{KPSW}$ and $L_{GLZ}$ (see Proposition~\ref{prop:kaltoNbpointsUniqueness} and
Theorem~\ref{thm:ISIT2}), depend on the bounds $N,D$ and $\tau$. Therefore, if
$N, D, \tau$ overestimate the degrees $\v$ and $d$ and the actual number of
errors, the corresponding $L_{KPSW}, L_{GLZ}$ would be too big compared to the
number we really need.  An approach to overcome this problem consists in the
introduction of an \textit{early termination} strategy whose goal is to decrease
the number of evaluations needed to recover a solution without knowing the
actual degrees of the solution and the number of errors. This strategy was first
proposed in \cite{kaltofen_early_2017} and it was based on the introduction of
some parameters that try to estimate the \textit{degrees of the solution} $\y(x)=\v(x)/d(x)$ of the PLS \eqref{eq:isitPLS}.  In this work, we
revisit this method (Proposition~\ref{prop:stopcritkalto}), by introducing some
parameters $\nu, \vartheta$ which represent attempts to find \textit{the actual degrees
of the key equation solution} $(\Lambda \v, \Lambda d)$, and a criterion which allows us to check if these
parameters $\nu, \vartheta$ upper bound these degrees. This makes our strategy sensitive to the \textit{actual number of errors} (instead of the upper bound $\tau$ which can be imprecise) and to the actual degrees of the solution $\y(x)$ (see Remark~\ref{rem:diffKalto}). 

Another significant difference from  \cite{kaltofen_early_2017} consists in the reduction of the number of evaluations which guarantees to uniquely recover the solution in presence of random errors.

We divide this section into two parts. First we assume to know a fixed upper bound on the number of errors that the nodes could make. We then notice that in the early termination algorithms (Algorithms~\ref{algo:ETKalto}, \ref{algo:ourEarlyTerm}, \ref{algo:ETLinearRS}, \ref{algo:ETLinearIRS}) the number of evaluations is iteratively incremented. The number of errors depends on this number of evaluations and so it can be hard to find a valid upper bound for it.
For this reason (as in \cite{kaltofen_early_2017}) in the second part, we introduce a \textit{linear error bound} which depends on the variable number of evaluations and on an error rate $\rho_E$.
In both cases, we propose two counting for $L$; one that can correct any error and one which derives from our Theorem \ref{thm:ISIT2} for the scenario where errors are random.
\subsection{Fixed error bound} \label{sec:fixed-error-bound}
We start by recalling and introducing some useful notations that we will use throughout this section.
Let $A(x)\y(x)=\b(x)$ be a PLS and $\y(x)=\frac{\v(x)}{d(x)}$ with $\gcd(\gcd_i(v_i),d)=1$ with $d$ monic. Consider $1\leq N,D\leq L$ such that $N>\deg(\v)$ and $D>\deg(d)$ and $\deg(A), \deg(\b)$.

In this first part of the section we also assume to know $\tau\geq |E|$ (see Section~\ref{sec:PLSwE}).
Let $\nu, \vartheta\geq 1$ and denote
\begin{equation}\label{eq:LLnutheta}
	\LL(\nu,\vartheta):=\min\{\max\{N-1+\vartheta, D-1+\nu\},\max\{\deg(A)+\nu, \deg(\b)+\vartheta\}\}.
\end{equation}

\subsubsection{Bounding $L$ for any error}
From \cite{kaltofen_early_2017} we can derive the following proposition, adapted to our choice of the parameters $\nu, \vartheta$ that gives a criterion for an early termination algorithm that can correct any error. 
\begin{proposition}\label{prop:stopcritkalto}
  Let $\vartheta, \nu, \geq 1$ and consider
  $L\geq \LL(\nu,\theta)+\tau$ evaluation points
  $\{\alpha_1, \ldots, \alpha_L\}$, where $\tau \ge |E|$. % is a bound on the number of errors.
  Fix $A(x)\y(x)=\b(x)$ and denote
  $\y(x) = \frac{\v(x)}{d(x)}$ with $\gcd(\gcd_i(v_i),d)=1$ and $d$
  monic. 
	
  Then the solution space of the key equations \eqref{eq:keyEqPLSWE} with input $Y, \nu, \vartheta$ is
  $
  \SS_{Y,\nu,\vartheta}=\langle x^i\Lambda \v, x^i\Lambda d\rangle_{0\leq i
    < \delta_{\nu, \vartheta}}
  $
  where
  $
  \delta_{\nu, \vartheta}=\min\{\nu-(\deg(\v)+|E|), \vartheta-(\deg(d)+|E|)\}.
  $
\end{proposition}

\begin{proof}
We now prove that
  $\SS_{Y,\nu,\vartheta} \subset \langle x^i\Lambda \v, x^i\Lambda
  d\rangle_{0\leq i < \delta_{\nu, \vartheta}}$, the other inclusion being
  straightforward. Let $(\bvarphi, \psi)\in \SS_{Y,\nu,\vartheta}$. For
  any $1\leq i \leq n$ and $1\leq j \leq L$, we have
  \begin{align}
    \varphi_i(\alpha_j)=y_{i,j}\psi(\alpha_j), \quad
    (\Lambda v_i) (\alpha_j) = y_{i,j} (\Lambda d) (\alpha_j) \label{eqb}
  \end{align}

  Assume that $L \ge \max\{D-1+\nu, N-1+\vartheta\} +\tau$.  If we multiply the first equation 
  in \eqref{eqb} by $(\Lambda d)(\alpha_j)$ and the second by $\psi(\alpha_j)$ and we subtract them, we get
  $
  (\Lambda (v_i \psi - d \varphi_i))(\alpha_j)=0
  $ for any $1\leq j \leq L$.  Now, since the polynomial
  $\Lambda (v_i \psi - d \varphi_i)$ has $L$ roots and degree
  $< |E| + \max(\deg(v) + \vartheta , \deg(d) + \nu )$ which is smaller than $L$
  by assumption,  so it is the zero polynomial.
  
  On the other hand, assume
  $L \ge \max\{\deg(A)+\nu, \deg(\b)+\vartheta\} + \tau$.  Since for any
  $j \notin E$, $\y_j=A(\alpha_j)^{-1}\b(\alpha_j)$. Combining this equation with $\bvarphi(\alpha_j)=\y_{j}\psi(\alpha_j)$ we get
  $(A\bvarphi - \psi\b)(\alpha_j) = 0$.  Now, notice that the vector of
  polynomials $A(x)\bvarphi(x)-\psi(x)\b(x)$ has $L - |E|$ roots and
  degree $< \max\{\deg(A)+\nu, \deg(\b)+ \vartheta\}$ which is smaller
  than $L-|E|$ by assumption. So $A(x)\bvarphi(x)
  =\psi(x)\b(x)$. Combined with $A(x) \v(x) = d(x) \b(x) $, we get
  $
  A(x)[\bvarphi(x)d(x)-\v(x)\psi(x)]=\boldsymbol{0}.
  $
  Since $A(x)$ is full rank, we obtain $\bvarphi(x)d(x)-\v(x)\psi(x)=\boldsymbol{0}$.

  Since  $\gcd(\gcd_i(v_i),d)=1$ and $d$ monic, there exists $P\in \Fq[x]$ such that
  $(\bvarphi, \psi)=(P\v, P d)$.

  Now $\bvarphi(\alpha_j)=\y_{j}\psi(\alpha_j)$ yields
  $(P(\v - \y d))(\alpha_j) = 0$ and so $P(\alpha_j) = 0$ for $j \in
  E$. This means that $\exists P'\in \Fq[x], P = \Lambda P'$. Finally,
  $(\bvarphi, \psi)=P'(\Lambda \v, \Lambda d)$ and the degree constraints
  on $(\bvarphi, \psi)$ imply $\deg P' < \delta_{\nu, \vartheta}$
  which concludes our proof.
\end{proof}

Proposition~\ref{prop:stopcritkalto} gives a criterion to check if
$(\nu,\vartheta)$ are upper bounds on the degree of the solution
$(\Lambda \v, \Lambda d)$. Indeed, $\nu > \deg(\v)+|E|$ and $ \vartheta > \deg(d)+|E|$)
\iff $\delta_{\nu, \vartheta} > 0$ \iff
$\SS_{Y,\nu,\vartheta} \neq\{(\boldsymbol{0},0)\}$. 

Let
\texttt{Check}$(L,\nu,\vartheta)$ be the function that returns the Boolean \linebreak
$\SS_{Y,\nu,\vartheta}\ \verb+!=+\ \{(\boldsymbol{0},0)\}$. If \texttt{Check}
returns \verb+true+, then we can call \linebreak
\texttt{FindSolution}$(Y,\nu,\vartheta)$
to recover $(\v, d)$ from $\SS_{Y,\nu,\vartheta}$ (see
Remark~\ref{remark:coeffMatrIRS}).  We are now ready to introduce
Algorithm~\ref{algo:ETKalto}, whose correctness follows from
Proposition~\ref{prop:stopcritkalto}.

\begin{algorithm}[t]
	\SetKwData{Left}{left}\SetKwData{This}{this}\SetKwData{Up}{up}
	\SetKwFunction{Union}{Union}\SetKwFunction{FindCompress}{FindCompress}
	\SetKwInOut{Input}{Input}\SetKwInOut{Output}{Output}%
	\Input{ a stream of vectors $Y=(\y_j)$ for $j= 1,2, \ldots$ which is
		extensible on demand, where
		$\y_j=\frac{\v(\alpha_j)}{d(\alpha_j)}+\e_j$, $N>\deg(\v)$, $D>\deg(d)$,
		$\tau \ge |E(L)|=|\{j \mid \e_j \neq \boldsymbol{0}\}|$, $\deg(A)$, $\deg(\b)$}
	\Output{$(\v,d)$ the solution of \eqref{eq:isitPLS} }
	\BlankLine
	
	% $\nu \leftarrow 1$, $\vartheta \leftarrow 1$, $\xi \leftarrow 0$\;
	$L \leftarrow \LL(1,1)+\tau$ \label{line:nbeva}\;
	\While{\texttt{true}}
	{
		\ForEach{$\nu, \vartheta$ with $\LL(\nu, \vartheta)+\tau==L$}{
			\If{\texttt{Check}$(L,\nu,\vartheta)$ \label{line:check2}}{
				\KwRet{\texttt{FindSolution}$((\y_j), L,\nu, \vartheta)$}\label{line:findSolution2}\;}
		}
		$L\leftarrow L+1$; Require a new $\y_j$\;
	}
\caption{Early Termination for PLSwE for fixed error bound $\tau$.}\label{algo:ETKalto}
\end{algorithm}
\DecMargin{1em} \IncMargin{1em}

Notice that \texttt{Check} and \texttt{FindSolution} perform the same
computation: they compute a basis of the module generated by solutions of the
key equation. Note also that in Algorithm~\ref{algo:ETKalto} the number of
evaluations varies, which could affect the number of errors. Therefore, we denote
$|E(L)|:=|\{1\leq j \leq L \mid \e_j\neq 0\}|$ instead of $|E|$ to stress out the dependency in $L$.

We now analyze the termination of Algorithm~\ref{algo:ETKalto}.

\begin{proposition}
  \label{prop:stopcritAlgo1}
  Algorithm~\ref{algo:ETKalto} terminates when $L$ equals
  $L^s := \LL(\deg(\v), \deg(d))+|E(L^s)|+1 + \tau$.
\end{proposition}
\begin{proof}
  We start by proving by contraposition that if \linebreak
  $\LL(\nu, \vartheta)+\tau<\LL(\deg(\v), \deg(d))+|E(L^s)|+1+\tau$ then
  $\delta_{\nu,\vartheta} \le 0$, \ie the algorithm do not stop (since in this case \texttt{Check} returns false).
  Indeed, if
  $\delta_{\nu,\vartheta} > 0$, then $\deg(\v)+|E(L^s)|< \nu$ and
  $\deg(d)+|E(L^s)|< \vartheta$ and so,
  $\LL(\nu, \vartheta)\geq \LL(\deg(\v), \deg(d))+|E(L^s)|+1$ (see
  \eqref{eq:LLnutheta}).

  Then, note that for $\nu=\deg(\v)+|E(L^s)|+1$ and
  $\vartheta=\deg(d)+|E(L^s)|+1$, the number of evaluations
  $\LL(\nu,\vartheta) + \tau$ equals to $L^s$ and the algorithm stops
  ($\delta_{\nu,\vartheta} > 0$).
\end{proof}

\begin{remark}\label{rem:diffKalto}
 Compared to the number of evaluations \linebreak
$\LL(\deg(\v), \deg(d)) + 1 + 2 \tau$ of \cite[Equations 5 and
9]{kaltofen_early_2017} (taking $R^\ast = 0$, omitting the rank
drops), we can lower the evaluation bound of Proposition~\ref{algo:ETKalto}  due to
a dependency on the real number of errors $|E(L^s)|$. To the best of our
knowledge, this dependency is original in the literature.
\end{remark}

\subsubsection{Bounding $L$ for random errors}

We now present an early termination strategy applied to the PLSwE problem (Algorithm~\ref{algo:ourEarlyTerm}) which allows us to further reduce the number of evaluation points compared to \cite{kaltofen_early_2017}. Notice that we slightly modify the structure of Algorithm~\ref{algo:ETKalto} considering 
$L:=\LL(\nu, \vartheta)+\left\lceil\frac{\tau}{n}\right\rceil$
evaluation points (step~\ref{line:nbeva}).

\LinesNotNumbered
\begin{algorithm}[t]
Same as Algorithm~\ref{algo:ETKalto} except \\
Line 1: $L \leftarrow \LL(1,1)+\left\lceil \frac{\tau}{n}\right\rceil$ \\
Line 3: $\nu, \vartheta$ with $\LL(\nu, \vartheta)+\left\lceil \frac{\tau}{n}\right\rceil==L$
  \caption{Early Termination for PLSwE for fixed error bound $\tau$ for random errors.}\label{algo:ourEarlyTerm}
\end{algorithm}
\LinesNumbered
\DecMargin{1em} \IncMargin{1em}

 Algorithm~\ref{algo:ourEarlyTerm} is based on the following new result.

\begin{theorem}{\label{thm:earlyTerm1}}
  Let $\vartheta, \nu\geq 1$, consider
  $L\geq \LL(\nu,\theta)+\left \lceil \frac{\tau}{n}\right \rceil$
distinct  evaluation points $\{\alpha_1, \ldots, \alpha_L\}$. Let
  $E\subseteq \{1, \ldots, L\}$.  Moreover, fix $A(x)\y(x)=\b(x)$ and denote
  $\y(x) = \frac{\v(x)}{d(x)}$ with $\gcd(\gcd_i(v_i), d)=1$ and $d$ monic.
  
  Consider the random matrix $Y$ where we denote the columns as $\y_j:=Y_{*,j}$, such that $\y_j$ is a uniformly distributed
  element of $\Fq^{n \times 1}$ if $j \in E$, and  $\y_j=\frac{\v(\alpha_j)}{d(\alpha_j)}$  if $j \notin E$.
  Then the solution space satisfies
  $
  \SS_{Y,\nu,\vartheta}=\langle x^i\Lambda \v, x^i\Lambda d\rangle_{0\leq i < \delta_{\nu, \vartheta}}
  $
  where
  $
  \delta_{\nu, \vartheta}=\min\{\nu-(\deg(\v)+|E|), \vartheta-(\deg(d)+|E|)\},
  $
  with probability at least $1-\frac{\vartheta}{q}$.
\end{theorem}
\begin{proof}
  The structure of the proof is the same as the proof of Theorem~\ref{thm:ISIT2}.
  In the first part we prove that there exists a draw of columns $\y_j$ for $j \in E$ for which the corresponding solution space $\SS_{Y,\nu,\vartheta}$  is generated by elements of the form $\langle x^i\Lambda \v, x^i \Lambda d\rangle$. More specifically recall that $\langle x^i\Lambda \v, x^i \Lambda d\rangle \subseteq \SS_{Y,\nu,\vartheta}$ is always true and so we need to prove the other inclusion in order to get the equality.
  In the second part, by using the Schwartz-Zippel Lemma we determine the bound on fraction of errors for which the solution space is not of the form $\langle x^i\Lambda \v, x^i \Lambda d\rangle$.
  
  Since here we are considering some general parameters $\nu, \vartheta$ instead
  of the bounds $N+\tau,D+\tau$, the only difference between this proof and the
  previous one consists in the first part. We still consider $E=\cup_{i=1}^n I_i$,
  such that for any $1\leq i \leq n$, $|I_i|\leq \lceil |E|/n\rceil$.  Recall
  that for any $j \in E$, we denote by $i_j$ the unique index such that
  $j \in I_{i_j}$.
  
  \smallskip We first assume that
  $\LL(\nu, \vartheta)=\max\{N-1+\vartheta, D-1+\nu\}$.  Construct a matrix $V$
  (as in the proof of Theorem~\ref{thm:ISIT2}) such that
  $V_{*,j}=\frac{\v(\alpha_j)}{d(\alpha_j)}$ if $j \notin E$, and
  $V_{*,j}\in \Fq^{n \times 1}$ is chosen so that
  $\v(\alpha_j)-d(\alpha_j)V_{*,j}=\eps_{i_j}$ if $j \in E$ ($\eps_{i}$ is the
  $i$th canonical vector).

  Let $(\bvarphi, \psi)\in \SS_{V,\nu,\vartheta}$. We now
  denote $\p := \psi\v-d\bvarphi$, and $p_i$ its $i$th component. We now show that
  $\p(x) = \bz$.

  We already have
  $\p(\alpha_j) = (\psi\v-d\bvarphi)(\alpha_j)=0$ for $j \notin E$. For $j \in E$, we can
  combine $\v(\alpha_j)-d(\alpha_j)V_{*,j}=\eps_{i_j}$ and
  $\bvarphi(\alpha_j)=V_{*,j}\psi(\alpha_j)$ to get
  $\p(\alpha_j)  =
  %(\psi\v-d \bvarphi)(\alpha_j) =
  \psi(\alpha_j)\v(\alpha_j)-d(\alpha_j)\psi(\alpha_j)V_{*,j}=\psi(\alpha_j)\eps_{i_j}.$
 
 % \begin{itemize}
 %  \item $V_{*,j}=\frac{\v(\alpha_j)}{d(\alpha_j)}$, if $j \notin E$,
 %  \item otherwise $V_{*,j}\in \Fq^{n \times 1}$ is chosen so that
 %    \begin{equation}{\label{eq:eTproof1}}
 %      \v(\alpha_j)-d(\alpha_j)V_{*,j}=\eps_{i_j}
 %    \end{equation}
 %    where $\eps_{i_j}$ is a vector of the canonical basis of $\Fq^{n \times 1}$.
 %  \end{itemize}

 %  By combining 
  
 %  We now consider $(\bvarphi, \psi)\in \SS_{V,\nu,\vartheta}$.
 %  By multiplying \eqref{eq:eTproof1} by $\psi(\alpha_j)$ and since $\bvarphi(\alpha_j)=V_{*,j}\psi(\alpha_j)$ for any $1\leq j \leq L$, then
 %  $$
 %  \psi(\alpha_j)\v(\alpha_j)-d(\alpha_j)\underbrace{\psi(\alpha_j)V_{*,j}}_{\bvarphi(\alpha_j)}=\psi(\alpha_j)\eps_{i_j}.
 %  $$
  Fix $1\leq i \leq n$, then for any $j \notin I_i$ then
  $p_i(\alpha_j)
  % = \psi(\alpha_j)v_i(\alpha_j)-d(\alpha_j)\varphi_i(\alpha_j)
  =0$. Now, notice
  that $p_i$ has degree $\leq \max\{\vartheta+N-1, D-1+\nu\}-1$
  and its number of roots is
  $L-|I_i| \geq L-\lceil|E|/n\rceil \geq \LL(\nu,\vartheta)-\lceil|\tau|/n\rceil=\max\{\nu+D-1,
  \vartheta+N-1\}$ and so it is the zero polynomial. Therefore,
  $\p(x)=\psi(x)\v(x)-d(x)\bvarphi(x) = \boldsymbol{0}$. The rest follows by observing that
  $\y(x)=\frac{\v(x)}{d(x)}$ is such that $\gcd(\gcd_i(v_i),d)=1$ and $d$ is monic. So, we can conclude that
  $\SS_{V,\nu,\vartheta}=\langle x^i\Lambda\v, x^i\Lambda d\rangle_{0\leq i <
    \delta_{\nu,\vartheta}}$.
  
  \smallskip We now assume that
  $\LL(\nu, \vartheta)=\max\{\deg(A)+\nu, \deg(\b)+\vartheta\}$.  Construct a
  matrix $V$ such that $V_{*,j}=\frac{\v(\alpha_j)}{d(\alpha_j)}$ if
  $j \notin E$, and $V_{*,j}\in \Fq^{n \times 1}$ is chosen so that
  $\v(\alpha_j)-d(\alpha_j)V_{*,j}=-A(\alpha_j)^{-1}d(\alpha_j)\eps_{i_j}.$ As in the proof of Theorem~\ref{thm:ISIT2}, for $j\in E$ then
  $\psi(\alpha_j)\eps_{i_j}=A(\alpha_j)\psi(\alpha_j) V_{*,j}
  -\b(\alpha_j)\psi(\alpha_j)=(A\bvarphi-\b\psi)(\alpha_j)$

  We now denote $\p:=A\bvarphi-\b\psi$ and by $p_i$ its $i$th component.  Fix
  $1\leq i \leq n$, then we observe that for any $j \notin I_i$ then
  $p_i(\alpha_j)=0$ (by the same argument of the proof of
  Theorem~\ref{thm:ISIT2}). Now, notice that
  $\deg(p_i)\leq \max\{\deg(A)+\nu, \deg(\b)+\vartheta\}-1$ and that the number
  of roots is $L-|I_i| \geq \max\{\deg(A)+\nu, \deg(\b)+\vartheta\}$ and so
  $p_i=0$. Therefore, $\p(x)=A(x)\bvarphi(x)-\b(x)\psi(x)=\boldsymbol{0}$.  The
  proof then follows by observing that $\y(x)=\frac{\v(x)}{d(x)}$ is the only
  solution of the linear system $A(x) \y(x) = \b(x)$ and is such that
  $\gcd(\gcd(v_i),d)=1$ with $d$ monic (with the same argument than the proof of
  Theorem~\ref{thm:ISIT2}).
\end{proof}

Therefore, under the assumptions of Theorem~\ref{thm:earlyTerm1}, with $L\geq \LL(\nu, \vartheta)+\left\lceil \frac{\tau}{n}\right \rceil$ evaluation points, then 
{\footnotesize
  \begin{equation}\label{eq:solSpaceEarly}
    \underbrace{\SS_{Y,\nu,\vartheta}  = \langle x^i\Lambda \v, x^i\Lambda d
      \rangle_{0\leq i < \delta_{\nu, \vartheta }}}_{\small \text{with probability at least } 1-\frac{\vartheta}{q}} \begin{cases}
      =\{(\boldsymbol{0}, 0)\} &\Longleftrightarrow \delta_{\nu, \vartheta }\leq 0\\
      \neq \{(\boldsymbol{0}, 0)\} &\Longleftrightarrow \delta_{\nu, \vartheta }> 0
    \end{cases}
  \end{equation}}

Therefore, the $\texttt{Check}(L,\nu,\vartheta)$ function remains unchanged but
it could return an incorrect answer with probability $\le \frac{\vartheta}{q}$.
On the other hand, the $\texttt{FindSolution}$ is slightly different from the
one that we have previously defined. Indeed, it could happen that the
$\Fq[x]$-module $\MM$ generated by solutions in $\SS_{Y,\nu\vartheta}$ has $\rank(\MM)>1$, in
which case the function outputs a failure message. Note that even if
$\rank(\MM)=1$, $\texttt{FindSolution}$ could return an incorrect output; for
instance if
$\SS_{Y,\nu,\vartheta} \neq \langle x^i\Lambda \v, x^i\Lambda d \rangle_{0\leq i
  < \delta_{\nu, \vartheta }} = \{(\bz,0)\}$. Remark that both these problems
can only happen if
$\SS_{Y,\nu,\vartheta} \neq \langle x^i\Lambda \v, x^i\Lambda d \rangle_{0\leq i
  < \delta_{\nu, \vartheta }}$, so they have probability
$\le \frac{\vartheta}{q}$.

We now better analyze how Algorithm~\ref{algo:ourEarlyTerm} works:
\begin{itemize}
\item if $\SS_{Y,\nu,\vartheta}=\{(\boldsymbol{0},0)\}$, by
  \eqref{eq:solSpaceEarly}, the parameters $\nu, \vartheta$ must be too small
  compared to $\deg(v)+|E|, \deg(d)+|E|$. Hence, in this case Algorithm~\ref{algo:ourEarlyTerm} increments the number of evaluations.
\item Otherwise, if $\SS_{Y,\nu,\vartheta}\neq \{(\boldsymbol{0},0)\}$ then
  \begin{itemize}
  \item by Theorem~\ref{thm:earlyTerm1}, with probability at least
    $1-\frac{\vartheta}{q}$ we have
    $\SS_{Y,\nu,\vartheta} = \langle x^i\Lambda \v, x^i\Lambda d \rangle_{0\leq
      i < \delta_{\nu, \vartheta}}$ which is nontrivial. By
    \eqref{eq:solSpaceEarly}, $\delta_{\nu, \vartheta}> 0$ and so the
    $\Fq[x]$-module generated by solutions of PLSwE is uniquely
    generated. Hence, the \linebreak
    $\texttt{FindSolution}$ function returns $(\v,d)$.
  \item On the other hand, with probability at most $\le \frac{\vartheta}{q}$ we
    have
    $\SS_{Y,\nu,\vartheta} \neq \langle x^i\Lambda \v, x^i\Lambda d
    \rangle_{0\leq i < \delta_{\nu, \vartheta}}$,
    % but we may have $\delta_{\nu, \vartheta}<0$
    and $\texttt{FindSolution}$ either returns a failure message or another
    solution $(\v', d')\neq (\v,d)$.
  \end{itemize}
\end{itemize}

\begin{remark}
  \label{rk:forEach2attempts}
  We can optimize the steps 2 and 3 of both Algorithms~\ref{algo:ETKalto} and
  \ref{algo:ourEarlyTerm} respectively by only testing two specific
  $(\nu,\vartheta)$ instead of all those giving a fixed $L$.  The goal is to
  make early termination algorithms have a smaller failure probability, and
  incidentally to make them faster.

 So we should try to find which $(\nu,\vartheta)$ maximize
  $\delta_{\nu, \vartheta}$ among those that give the same number of evaluations
  $L$. Indeed, our goal is to have $\delta_{\nu, \vartheta} > 0$. The two
  candidates are $(\nu_1,\vartheta_1) = (\LL-(D-1),\LL-(N-1))$ and
  $(\nu_2,\vartheta_2) = (\LL-\deg(A),\LL-\deg(\b))$ where
  $\LL := L - \ceil{\frac{\tau}{n}}$.

  We now show that any $(\nu,\vartheta)$ such that $\LL = \LL(\nu,\vartheta)$, we
  have
  $(\nu,\vartheta) \le (\nu_1,\vartheta_1) \textrm{ or } (\nu,\vartheta) \le
  (\nu_2,\vartheta_2)$ (for the partial component-wise ordering). Indeed, either
  $\LL = \LL(\nu,\vartheta) = max(D-1+\nu,N-1+\vartheta)$ and then
  $(\nu,\vartheta) \le (\nu_1,\vartheta_1)$, or
  $\LL = \LL(\nu,\vartheta) = max(\deg(A)+\nu,\deg(\b)+\vartheta)$ and then
  $(\nu,\vartheta) \le (\nu_2,\vartheta_2)$.
  
  Remark that if $D-1\le\deg(A), N-1 \le \deg(\b)$ then we should only try
  $(\nu_1,\vartheta_1)$ because $(\nu_1,\vartheta_1) \ge (\nu_2,\vartheta_2)$
  and $\LL(\nu_1,\vartheta_1) = \LL$. Similarly if
  $D-1\ge\deg(A), N-1 \ge \deg(\b)$ then we should only try
  $(\nu_2,\vartheta_2)$. However, if $D-1 \le \deg(A), N-1 \ge \deg(\b)$, we
  should try both candidates because they are not comparable and they both lead
  to $\LL(\nu_1,\vartheta_1) = \LL(\nu_2,\vartheta_2) =\LL$, and as a
  consequence of the same number of evaluations.
\end{remark}

\begin{proposition} \label{prop:LstopIRSFixed}
  Algorithm~\ref{algo:ourEarlyTerm} terminates with at most $L^{s}$ evaluation points,
  where $L^{s}$ is such that \linebreak
  $L^{s}=\LL(\deg(\v),\deg(d))+|E(L^{s})|+1+\left\lceil\frac{\tau}{n}\right\rceil$. Its output is correct
  with probability at least $1- \frac{2(D+E(L^{s})+1)(\deg(\Lambda\v, \Lambda d) + 1)}{q}$.
\end{proposition}

\begin{proof} 
  Assume first that we are always in the favorable cases, \ie
  $\SS_{Y,\nu,\vartheta}=\langle x^i\Lambda \v, x^i\Lambda d\rangle_{0\leq i <
    \delta_{\nu, \vartheta}}$ for any attempt $(\nu,\vartheta)$.  In other terms
  we assume that Algorithm~\ref{algo:ourEarlyTerm} returns the solution
  $(\Lambda \v, \Lambda d)$ of the PLS \eqref{eq:isitPLS}. Then the proof that
  Algorithm~\ref{algo:ourEarlyTerm} stops with exactly $L^s$ evaluations is
  similar to the proof of Proposition~\ref{prop:stopcritAlgo1}.

  We now analyze the unfortunate cases in order to bound the corresponding
  probability:
  \begin{enumerate}
  \item it could happen that for a certain attempt $(\nu,\vartheta)$, then
    $\SS_{Y,\nu,\vartheta} \neq \{(\boldsymbol{0}, 0)\}$ and
    $\delta_{\nu, \vartheta} \le 0$. In this case the \texttt{Check} function at
    step 4 would return \verb+false+. Then Algorithm~\ref{algo:ourEarlyTerm}
    would stop prematurely, \ie before reaching $L^s$ evaluations, with a
    failure message or an incorrect output.

  \item it could also happen that Algorithm~\ref{algo:ourEarlyTerm} reaches
    $L^s$ evaluations and parameters $(\nu,\vartheta)$ such that
    $\delta_{Y,\nu,\vartheta} > 0$ but \texttt{FindSolution} returns an
    incorrect PLS solution. Note first that \texttt{Check} must return
    \textsf{true} because
    $\SS_{Y,\nu,\vartheta} \supseteq \langle x^i\Lambda \v, x^i\Lambda
    d\rangle_{0\leq i < \delta_{\nu, \vartheta}} \neq \{(\bz, 0)\}$. However, if
    $\SS_{Y,\nu,\vartheta} \neq \langle x^i\Lambda \v, x^i\Lambda
    d\rangle_{0\leq i < \delta_{\nu, \vartheta}}$ then \texttt{FindSolution}
    will return a failure message or an incorrect output.
  \end{enumerate}

  The probability of falling in an unfortunate case is related to the number of
  $(\nu,\vartheta)$ we attempt before reaching
  $L^s=\LL(\deg(\v)+|E(L^{s})|+1,
  \deg(d)+|E(L^{s})|+1)+\left\lceil\frac{\tau}{n}\right\rceil$. The number of
  evaluations starts from $\LL(1,1) + \ceil{\frac{\tau}{n}}$ and ends at $L^s$:
  there are at most
  $\max(\deg(\v),\deg(d)) + |E(L^{s})| + 1=\deg(\Lambda \v, \Lambda d)+1$
  different evaluation counting.
  Recall that in order to minimize computations, for each evaluation counting, we
  try at most 2 affectations of parameters $(\nu_1,\vartheta_1)$ and
  $(\nu_2,\vartheta_2)$ (see Remark~\ref{rk:forEach2attempts}).
  Now any attempt $(\nu,\vartheta)$ could fail with probability
  $\le \frac{\vartheta}{q}$. It remains to upper bound $\vartheta$ among all
  attempts.
  
  Denote by $(\nu^s_1,\vartheta^s_1)$, $(\nu_2^s,\vartheta_2^s)$ the candidate
  parameters corresponding to $L^{s}$.
  
  Therefore,
  $\max\{\vartheta_1^s, \vartheta_2^s\}=\LL(\deg(\v),
  \deg(d))+E(L^{s})+1-\min\{N-1,\deg(\b)\}$. First if $N-1\leq \deg(\b)$, then
  since \linebreak
  $\LL(\deg(\v), \deg(d))+E(L^{s})+1\leq
  \max\{D-1+\deg(\v),N-1+\deg(d)\}+E(L^{s})+1 \leq N+D+E(L^{s})$, we have
  $\max\{\vartheta_1^s, \vartheta_2^s\}\leq D+E(L^{s})+1$.

  Otherwise, $N-1 \geq \deg(\b)$, and we have
  $\LL(\deg(\v), \deg(d))+E(L^{s})+1\leq \max\{\deg(A)+\deg(\v),
  \deg(\b)+\deg(d)\}+E(L^{s})+1=\deg(\b)+\deg(d)+E(L^{s})+1$ (since
  $A\v = \b d$) then
  $\max\{\vartheta_1^s, \vartheta_2^s\}\leq \deg(d)+1+E(L^{s})\leq D+E(L^{s})+1.
  $ Combining all these results we can conclude that the probability to reach
  the unfortunate cases is at most
  $\frac{2(D+E(L^{s})+1)(\deg(\Lambda\v, \Lambda d) + 1)}{q}$.
\end{proof}

\subsection{Linear error bound}\label{sec:linearBound}

Up until now, our early termination schemes have assumed that the number of
errors was bounded by a constant $\tau$. Since early termination requires more
and more evaluations, it would be interesting to have an error bound that
depends on the number of evaluation.  In this
section, 
we consider a linear error bound which depends on an \textit{error rate}.

\begin{assumption}\label{ass:nberror}
  For any number of evaluation points $L$, the number of errors $E(L)$ is
  bounded by $|E(L)| \le \rho_E L$ where $0 \le \rho_E < 1/2$.
 
\end{assumption}

This assumption comes from
\cite{kaltofen_sparse_2013,kaltofen_sparse_2014,kaltofen_early_2017}, where they
also consider the variant $|E(L)| \le \ceil{\rho_E L}$. For the sake of
simplicity, we restrict ourselves to Assumption~\ref{ass:nberror}. However, our results can be adapted to the alternative linear error bound.

\subsubsection{Bounding $L$ for any error}
We start by adapting Proposition~\ref{prop:stopcritkalto} to the special case of
linear error bound.

\begin{proposition}
  \label{prop:LinearDetectRS}
  Let $\nu, \vartheta \geq 1$ and consider
  $L = \floor{\frac{\LL(\nu,\vartheta)+1}{1-\rho_E}}$ evaluation points and
  error bound $\tau = \floor{\rho_E L}$.
  Fix $A(x)\y(x)=\b(x)$ and denote $\y(x) = \frac{\v(x)}{d(x)}$ with
  $\gcd(\gcd_i(v_i),d)=1$ and $d$ monic.
  
  Under Assumption~\ref{ass:nberror}, we have
  $  \SS_{Y,\nu,\vartheta}=\langle x^i\Lambda \v, x^i\Lambda d\rangle_{0\leq i <
    \delta_{\nu, \vartheta}}.
  $
\end{proposition}

\begin{proof}
  We use the following notations
  $\bar{L}^\ast(\nu,\vartheta) = \frac{\LL(\nu,\vartheta)+1}{1-\rho_E}$, so that
  $L^\ast(\nu,\vartheta) = \floor{\bar{L}^\ast(\nu,\vartheta)}$, and
  $\tau^\ast(\nu,\vartheta) = \floor{\rho_E L^\ast(\nu,\vartheta)}$.

  We first show that we are under the hypotheses of
  Proposition~\ref{prop:stopcritkalto}, \ie that
  $L^\ast(\nu,\vartheta) \ge \LL(\nu,\vartheta) + \tau^\ast(\nu,\vartheta)$ and
  $\tau^\ast(\nu,\vartheta) \ge |E(L^\ast(\nu,\vartheta)|$.
  
  First $|E(L^\ast(\nu,\vartheta))| \le \rho_E L^\ast(\nu,\vartheta)$ using
  Assumption~\ref{ass:nberror}. Since $|E(L^\ast(\nu,\vartheta))| \in \N$, we
  get 
  $|E(L^\ast(\nu,\vartheta))| \le \floor{\rho_E L^\ast(\nu,\vartheta)} =
  \tau^\ast(\nu,\vartheta)$.  
  Then
  \begin{eqnarray*}
    \LL(\nu,\vartheta) + \tau^\ast(\nu,\vartheta)
      & \le & \LL(\nu,\vartheta) + \rho_E L^\ast(\nu,\vartheta) \\
      & \le & \LL(\nu,\vartheta) + \rho_E \bar{L}^\ast(\nu,\vartheta) \\
      & = & \bar{L}^\ast(\nu,\vartheta) - 1,
  \end{eqnarray*}
  and so
  $\LL(\nu,\vartheta) + \tau^\ast(\nu,\vartheta) \le
  \floor{\bar{L}^\ast(\nu,\vartheta)} = L^\ast(\nu,\vartheta)$.
\end{proof}

Therefore, we can use the evaluation counting
$L(\nu,\vartheta) = \floor{\frac{\LL(\nu,\vartheta)+1}{1-\rho_E}}$ to detect if
$(\nu,\vartheta)$ are good estimations, and eventually return the solution
$(\v,d)$ of the PLS. This is exactly what Algorithm~\ref{algo:ETLinearRS} does
and Proposition~\ref{prop:LinearDetectRS} shows its correction.

\begin{algorithm}[t]
  \SetKwInOut{Input}{Input}\SetKwInOut{Output}{Output}%
  \Input{ a stream of vectors $Y=(\y_j)$ for $j= 1,2, \ldots$ which is
    extensible on demand, where
    $\y_j=\frac{\v(\alpha_j)}{d(\alpha_j)}+\e_j$ $N>\deg(\v)$, $D>\deg(d)$, $\deg(A)$, $\deg(\b)$,\\
    $0 \le \rho_E < 1/2$ s.t. $|E(L)| \le \rho_E L$ for all $L$. }
  \Output{$(\v,d)$ the solution of \eqref{eq:isitPLS} }
  \BlankLine
  
  $L^{num} \leftarrow \LL(1,1)+1$\;
  \While{\texttt{true}}
  {
    $L \leftarrow \floor{\frac{L^{num}}{1-\rho_E}}$;
    $\tau \leftarrow \floor{\rho_E L}$; Require new $\y_j$\;
    \ForEach{$\nu, \vartheta$ with $\LL(\nu, \vartheta)+1==L^{num}$}{
      \If{\texttt{Check}$(L,\nu,\vartheta)$}{
        \KwRet{\texttt{FindSolution}$((\y_j), L,\nu, \vartheta)$}\;}
    }
    $L^{num}\leftarrow L^{num}+1$\;
  }
  
  \caption{Early Termination for PLSwE for linear error bound $\rho_E$.}\label{algo:ETLinearRS}
\end{algorithm}

We now prove that Algorithm~\ref{algo:ETLinearRS}
stops with a certain number of evaluation points.

\begin{proposition} \label{prop:LstopLinearIRS}
\
  \begin{enumerate}
  \item Algorithm~\ref{algo:ETLinearRS} terminates.
  \item Algorithm~\ref{algo:ETLinearRS} stops when it reaches $L^{s}$ evaluation
    points where
    \begin{equation}
      \label{eq:LstopRSLinear}
      L^{s} = \floor{\frac{\LL(\deg(\v),\deg(d)) + |E(L^{s})| + 2}{1-\rho_E}}
    \end{equation}

  \item \label{item:compKalto} We can bound $L^{s} \le \floor{\frac{\LL(\deg(\v),\deg(d)) + 2}{1-2\rho_E}}$.
  \item \label{item:adaptiveCount}If for some reason fewer errors are made, \ie$|E(L)| \le \rho_E' L$
    with $\rho_E' < \rho_E$, then
    $L^{s} \le \floor{\frac{\LL(\deg(\v),\deg(d)) + 2}{1-\rho_E'-\rho_E}}$.
  \end{enumerate}
\end{proposition}

The inequality given in Item~\ref{item:compKalto} relates the performance of our
early termination algorithm to the literature. Indeed, the right-hand bound
can be derived from~\cite[Algorithm 2.2]{kaltofen_early_2017} with $\rho_R = 0$ (no
rank drop) and $q_R = q_E = +\infty$ (for simplicity).

Note that $\rho_E$ and $\rho_E'$ don't play the same role: $\rho_E$ must be
known in advance (it is an input of the algorithm) and be related to a linear
error bound that is always true. If Assumption~\ref{ass:nberror} fails then the
correctness of Algorithm~\ref{algo:ETLinearRS} may be lost. On the other hand,
$\rho_E'$ is used to demonstrate that our early termination technique is
sensitive to the real number of errors (in addition to real degrees of $\v,d$),
\ie that it can stop earlier if fewer errors than expected are made.

\begin{proof} We keep the notations of the proof of
  Proposition~\ref{prop:LinearDetectRS}.
  
  (1) We need to prove that Algorithm~\ref{algo:ETLinearRS} stops. Consider
  $f(x) := x - \frac{\LL(\deg(\v),\deg(d)) + \rho_E x + 2}{1-\rho_E}$, which is
  strictly increasing because $1 > \rho_E/(1-\rho_E)$. This last inequality is true since
  $0 \le \rho_E < 1/2$.

  Let $L \in \N$ such that $f(L) \ge 0$. Set
  $\nu' = \deg(\v) + |E(L)| + 1, \vartheta' = \deg(d) + |E(L)| + 1$ and
  $L' = L^\ast(\nu',\vartheta')$. We have
  \begin{eqnarray*}
    L' & = &\floor{\frac{\LL(\deg(\v),\deg(d)) + |E(L)| + 2}{1-\rho_E}} \\
    & \le & \frac{\LL(\deg(\v),\deg(d)) + \rho_E L + 2}{1-\rho_E} \le L    
  \end{eqnarray*}
  where the last inequality comes from $f(L) \ge 0$. Hence,
  $\nu' \ge \deg(\v) + |E(L')| + 1$, $\vartheta'\geq \deg(d)+|E(L')|+1$ and so
  $\delta_{\nu',\vartheta'} > 0$. So Algorithm~\ref{algo:ETLinearRS} would stop
  with $\leq L'$ evaluations.
 
  (2) Let $L^{s}$ be the number of evaluations when the algorithm stops. We now
  prove Equation~(\ref{eq:LstopRSLinear}). There exists $\nu^s, \vartheta^s$
  such that $L^{s} = L^\ast(\nu^s,\vartheta^s)$, and we must have
  $\delta_{\nu^s,\vartheta^s} > 0$, \ie$\nu > \deg(\v) + |E(L^{s})|$ and
  $\vartheta > \deg(d) + |E(L^{s})|$. Define
  $\nu' = \deg(\v) + |E(L^{s})| + 1, \vartheta' = \deg(d) + |E(L^{s})| + 1$. We
  now prove that $L^{s} = L^\ast(\nu',\vartheta')$ by contradiction, which
  implies Equation~(\ref{eq:LstopRSLinear}). So assume that
  $L^{s} > L' := L^\ast(\nu',\vartheta')$ (note that the inequality $\ge$ is
  always true since $(\nu, \vartheta) \ge (\nu',\vartheta')$ and
  $L^\ast(\nu,\vartheta)$ is increasing). But
  $\nu' = \deg(\v) + |E(L^{s})| + 1 > \deg(\v) + |E(L')|$ and 
  $\vartheta'= \deg(d) + |E(L^{s})| + 1 > \deg(d) + |E(L')|$ so that $\delta_{\nu',\vartheta'} > 0$ and
  Algorithm~\ref{algo:ETLinearRS} would have stopped with $L' < L^{s}$
  evaluations which is a contradiction.

  (3) Now let $\bar{L}' = \frac{\LL(\deg(\v),\deg(d))+2}{1-2\rho_E}$ and
  $L' = \floor{\bar{L}'}$. $\bar{L}'$ is defined so that $0 = f(\bar{L}')$. We
  now prove that $f(L^{s}) \le 0$, which implies $L^{s} \le \bar{L}'$ since $f$
  is strictly increasing, thus $L^{s} \le \floor{\bar{L}'} = L'$ since
  $L^{s} \in \N$. The claim comes from
  $L^{s} \le \frac{\LL(\deg(\v),\deg(d)) + |E(L^{s})| + 2}{1-\rho_E} \le
  \frac{\LL(\deg(\v),\deg(d)) + \rho_E L^{s} + 2}{1-\rho_E}$.

  (4) If one execution of PLSwE satisfies $|E(L)| \le \rho_E' L$ for
  $\rho_E' < \rho_E$, we can prove that
  $L^{s} \le \floor{\frac{\LL(\deg(\v),\deg(d))+2}{1-\rho_E' - \rho_E}}$ by
  adapting the previous proof with
  $\overline{f}(x) := x - \frac{\LL(\deg(\v),\deg(d)) + \rho_E' x+ 2}{1-\rho_E}$.
\end{proof}

\subsubsection{Bounding $L$ for random errors} As before, we can lower the number of evaluation points considering randomly distributed errors. Theorem~\ref{thm:earlyTerm1} can be
adapted to the context of a linear error bound.

\begin{proposition}  \label{prop:LinearDetectIRS}
  Under Assumption~\ref{ass:nberror} and assumptions of
  Theorem~\ref{thm:earlyTerm1}, and using
  $L = \floor{\frac{\LL(\nu,\vartheta)+1}{1-\rho_E/n}}$ evaluation points and error
  bound $\tau = \floor{\rho_E L}$ for any $\nu, \vartheta \ge 1$, we have \linebreak
  $
  \SS_{Y,\nu,\vartheta}=\langle x^i\Lambda \v, x^i\Lambda d\rangle_{0\leq i
    < \delta_{\nu, \vartheta}}
  $
  with probability at least $1-\frac{\vartheta}{q}$.
\end{proposition}

The proof is similar to the one of Proposition~\ref{prop:LinearDetectRS}, except
that it is based on Theorem~\ref{thm:earlyTerm1}.

With the help of the latter proposition, we can introduce
Algorithm~\ref{algo:ETLinearIRS}, which returns a correct solution of the PLSwE
problem with high probability.

\LinesNotNumbered
\begin{algorithm}[t]\label{algo:ourLinear}
Same as Algorithm~\ref{algo:ETLinearRS} except \\
Line 3: $L \leftarrow \floor{\frac{L^{num}}{1-\rho_E/n}}$;
    $\tau \leftarrow \floor{\rho_E L}$; Require new $\y_j$\;
  
  \caption{Early Termination for PLSwE for linear error bound $\rho_E$ for random errors.}\label{algo:ETLinearIRS}
\end{algorithm}
\LinesNumbered

\begin{proposition}

  Algorithm~\ref{algo:ETLinearIRS} stops with at most $L^{s}$ evaluations,
  where
  \begin{equation}
    \label{eq:LstopIRSXi}
    \small
    L^{s} = \floor{\frac{\LL(\deg(\v),\deg(d)) + |E(L^{s})| + 2}{1-\rho_E/n}} \le \floor{\frac{\LL(\deg(\v),\deg(d)) + 2}{1-(1+1/n)\rho_E}}
    % L^{s} = min \left \{ L^\ast \ \middle | \ L^\ast \ge \frac{\LL(\deg(\v),\deg(d)) +
    %   E(L^\ast) + 2}{1-\rho_E / n} \right \}.
  \end{equation}
  Its output is correct with probability
  $\ge 1- \frac{2 (D+E(L^s)+1)(\deg(\Lambda\v, \Lambda d) + 1)}{q}$.

  If for some reason fewer errors are made, \ie$|E(L)| \le \rho_E' L$ with
  $\rho_E' < \rho_E$, then
  $L^{s} \le \floor{\frac{\LL(\deg(\v),\deg(d)) + 2}{1-\rho_E'-\rho_E/n}}$.
\end{proposition}

\begin{proof}
  The proof concerning all the numbers of evaluation points is similar to the one of
  Proposition~\ref{prop:LstopLinearIRS}, but it is based on the number of evaluations of
  Theorem~\ref{thm:earlyTerm1}. The statement of correctness can be proved similarly
  to Proposition~\ref{prop:LstopIRSFixed}.
\end{proof}

\bibliographystyle{alpha}
\bibliography{main.bib}

\end{document}